\documentclass[amsfonts, amssymb, amsmath,longbibliography,nofootinbib,reqno]{amsart}
\usepackage{graphicx} % Required for inserting images
\usepackage{amsaddr}

\usepackage[dvipsnames]{xcolor}
\usepackage{float}
\usepackage[shortlabels]{enumitem}
\usepackage[numbers,sort&compress]{natbib}
\usepackage{braket}
\usepackage{amsthm}
\usepackage{mathtools}
\usepackage{url}
\usepackage{physics}
\usepackage{graphicx}
\usepackage[left=22mm,right=22mm,top=35mm]{geometry} 
\usepackage[T1]{fontenc}
\usepackage{bm}
\usepackage{tabularx}
\newcommand*{\Ad}{\mathrm{Ad}}

\def\HC{\mathcal{H}}

\def\ad{^{\dagger}}

\newcommand{\fsnull}[1]{}
\newcommand{\old}[1]{}

\usepackage[makeroom]{cancel}
\usepackage[toc,page]{appendix}
\definecolor{C1}{RGB}{52, 89, 149}
\definecolor{C2}{RGB}{251, 77, 61}
\definecolor{C3}{RGB}{3, 206, 164}
\definecolor{C4}{RGB}{202, 21, 81}
\usepackage{hyperref}
\hypersetup{colorlinks=true, linkcolor=C4, citecolor=C4, urlcolor=C4}

\usepackage{tikz}
\tikzset{every picture/.style=remember picture}

\usepackage{pgfplots}

\usepackage[utf8]{inputenc}
\usepackage{graphicx}
\usepackage{xcolor}
\usepackage{amsmath}
\usepackage{amsthm}
\usepackage{bm}
\usepackage{bbm}
\usepackage{comment}
\usepackage{appendix}
\usepackage{mathdots}
\usepackage{lipsum}
\usepackage{verbatim}
\usepackage{nccmath}
\usepackage{amsfonts}
\usepackage{thm-restate}
\usepackage{thmtools}
\usepackage{capt-of}
\usepackage{booktabs}

\newtheorem{remark}{Remark}
\newtheorem{crl}{Corollary}

\usepackage{amssymb}
\usepackage{dsfont}
\newcommand{\HS}{\text{HS}}

\newcommand{\supp}{\text{supp}}

\renewcommand{\geq}{\geqslant}
\renewcommand{\leq}{\leqslant}

\newcommand{\ot}{\otimes}
\newcommand{\ts}{^{\otimes 2}}
\newcommand{\bs}{\textsf{BS}}

\newcommand{\Lm}{\Lambda }

\newcommand{\sg}{\sigma }

\newcommand{\Om}{\Omega }

\DeclareMathOperator*{\expect}{\mathbb{E}}

\newcommand{\mcl}{\mathcal{L}}

\newcommand{\mco}{\mathcal{O}}

\newcommand{\mch}{\mathcal{H}}

\newcommand{\mcp}{\mathcal{P}}

\newcommand{\mce}{\mathcal{E}}

\newcommand{\mbsp}{\mathbb{SP}}

\newcommand{\mbo}{\mathbb{O}}
\newcommand{\mbu}{\mathbb{U}}

\def\be{\begin{equation}}
\def\ee{\end{equation}}
\def\bs{\begin{split}}
\def\e{\end{split}}
\def\ba{\begin{eqnarray}}
\def\bea{\begin{eqnarray}}

\def\tea{\end{eqnarray}}
\def\ea{\end{eqnarray}}
\def\eea{\end{eqnarray}}

\def\tk{^\otimes k}

\def\tk{^{\otimes k}}

\newcommand{\id}{\mathds{1}}

\def\sg{\sigma}

\def\be{\begin{equation}}
\def\te{\end{equation}}
\def\ee{\end{equation}}
\def\ba{\begin{eqnarray}}
\def\bea{\begin{eqnarray}}

\def\tea{\end{eqnarray}}
\def\ea{\end{eqnarray}}
\def\eea{\end{eqnarray}}

\newtheorem{theorem}{Theorem}

\begin{document}

\bibliographystyle{IEEEtran}

\title[Ambient unitaries don't enable shallow group designs]{Ambient unitaries don't enable   shallow group designs}
\author{\vspace{-3mm}M\MakeLowercase{axwell}  W\MakeLowercase{est}\,\textsuperscript{1,2},  M. C\MakeLowercase{erezo}\,\textsuperscript{2,3} \MakeLowercase{and}  M\MakeLowercase{art\'{i}n} L\MakeLowercase{arocca}\,\textsuperscript{1,2}}
\address{\textsuperscript{1}Theoretical Division, Los Alamos National Laboratory, Los Alamos, New Mexico 87545, USA}
\address{\textsuperscript{2}Quantum Science Center, Oak Ridge, TN 37931, USA}
\address{\textsuperscript{3}Information Sciences, Los Alamos National Laboratory, Los Alamos, New Mexico 87545, USA}

\begin{abstract} { 
    Characterising the efficiency with which    designs over various subsets of the unitary group may be constructed is an important goal of quantum information theory. While it is now known that approximate unitary designs can be realised in depth  logarithmic in the system size, it has recently been shown that ensembles of    local nearest-neighbour sublinear-depth one-dimensional circuits over the matchgate, orthogonal, and  symplectic  groups 
    cannot form approximate 2-designs over their parent groups; similarly, sublinear-depth ensembles of Cliffords cannot form a Clifford 4-design. In this note we show that this remarkable   exponential separation is not merely an artefact of restricting to ensembles consisting of unitaries from the subgroups themselves, but rather that   \textit{no} ensemble of local nearest-neighbour sublinear-depth unitaries can realise approximate designs in the aforementioned cases, even when employing ``ambient'' unitaries from beyond the subgroup itself (possibly acting on ancilla qubits). This implies that various natural tomography and benchmarking schemes which involves   sampling from these groups suffer from a   dramatic circuit depth overhead compared to similar protocols which involve sampling from the full   unitary group. We additionally conclude   that, in all of the above cases, the known linear-depth design constructions are up to constant factors optimal. 
}
\end{abstract}

\maketitle

\tableofcontents

\section{Introduction}
Understanding the complexity of  sampling from various random ensembles of unitaries  is a fundamental task  in quantum computation and information~\cite{schuster2024random,cui2025unitary,laracuente2024approximate,west2025no,grevink2025will,brandao2016local,haferkamp2022random,ma2024how,dankert2009exact,harrow2009random,chen2024incompressibility,gross2007evenly,haah2025short}, with applications ranging from practical near-term quantum algorithms to problems in quantum gravity~\cite{hayden2007black,sekino2008fast}. Naturally, the greatest attention has been given to constructing designs over the unitary group itself, where progress has been remarkable. Indeed, it has been shown that ensembles of depth scaling just as $\mco(k\,{\rm polylog}\, (k) \log(n/\varepsilon)) $ can form $\varepsilon$-approximate\footnote{We defer for the moment commentary on the metric by which we quantify approximations}  $k$-designs~\cite{schuster2024random}. When allowing ancilla qubits, long range gates, and assuming a supply of $\mco(nk) $ bits of randomness, this scaling can be further reduced  to $\mco(\log k \log\log nk/\varepsilon)$~\cite{cui2025unitary}. \\

Recently, there has also been interest in considering the construction of designs over strict subgroups of the unitary group~\cite{west2025no,grevink2025will}. In a surprising contrast to the extreme efficiency with which approximate designs can be achieved in the full unitary case, it has been shown~\cite{west2025no,grevink2025will} that sublinear-depth ensembles over the matchgate, orthogonal, and unitary symplectic groups cannot form 2-designs over their parent groups, and that sublinear-depth ensembles of Cliffords cannot form a Clifford 4-design. These results admit a simple explanation: the aforementioned groups possess invariant states in certain tensor products of their defining representations, which may be used as ``probes'' to detect sublinearity\footnote{The situation is marginally more subtle in the Clifford case, as we explore below}. Indeed, suppose for the sake of illustration that a state $\ket \Psi\in\mch\ts$ is invariant under the tensor square action of a group $G\subseteq \mbu(\mch)$ of unitaries on a Hilbert space $\mch$; that is, $U\ts\ket\Psi=\ket\Psi$ for all $U\in G$. Then, for any perturbation $V\in \mbu(\mch)$, we have $U\ts (V\ot\id_\mch)\ket\Psi =   ((UVU\ad)\ot\id_\mch) \ket\Psi $. The appearance of the Heisenberg-evolved $UVU\ad$ may then be used to probe the depth of $U$; the resulting  lightcone arguments lead to $\Om(n)$ lower bounds on the depths required for approximate designs. This is depicted in Figure~\ref{fig:1}. \\

The preceding argument, however, relies crucially on the assumption that we are forming a design by sampling from an ensemble of unitaries which themselves belong to the group under consideration; this is a significant assumption which is not necessarily true in practical design constructions. For example, let $\mce$ be any   $\varepsilon$-approximate Clifford 3-design and  $V$ some arbitrary non-Clifford unitary. Since the Cliffords are a unitary 3-design~\cite{webb2016clifford,zhu2017multiqubit}, the left-invariance of the unitary Haar measure implies that $V\mce$  (i.e., the ensemble $\{VU\,|\,U\sim\mce\}$) is again an $\varepsilon$-approximate Clifford 3-design, despite being supported  outside of the Clifford group. Cosets of this form have arisen   in randomised benchmarking; for example, $T\mathsf{Cl}_n$ (with $T$ the usual $T$-gate and $\mathsf{Cl}_n$ the $n$-qubit Clifford group)  in Ref.~\cite{harper2017estimating}. 
So, for the practical question of whether we can find ensembles of sublinear-depth unitaries which constitute      designs over the aforementioned groups, a considerable loophole remains open.  \\

Here we close this loophole for all four of our considered  cases, which turns out to be possible as a result of a simple observation. Indeed, consider  a ``typical'' $U\sim\mce$, where $\mce$ is an ensemble of shallow unitaries that may or may not have support solely within a target group $G$. There are two possibilities: either $U\tk$ approximately stabilises the $G$-invariant state $\ket\Psi$, or it does not. In the former case, we can effectively run the experiment of Ref.~\cite{west2025no} sketched above (thereby distinguishing $U$ from a typical element sampled from the Haar measure $\mu_G$ on $G$); in the latter case we can plainly distinguish it from the action of any $g\sim\mu_G$, which by assumption does not disturb $\ket\Psi$.  We make this intuition precise in the next section.

\begin{figure}
    \centering
    \includegraphics[width=\linewidth]{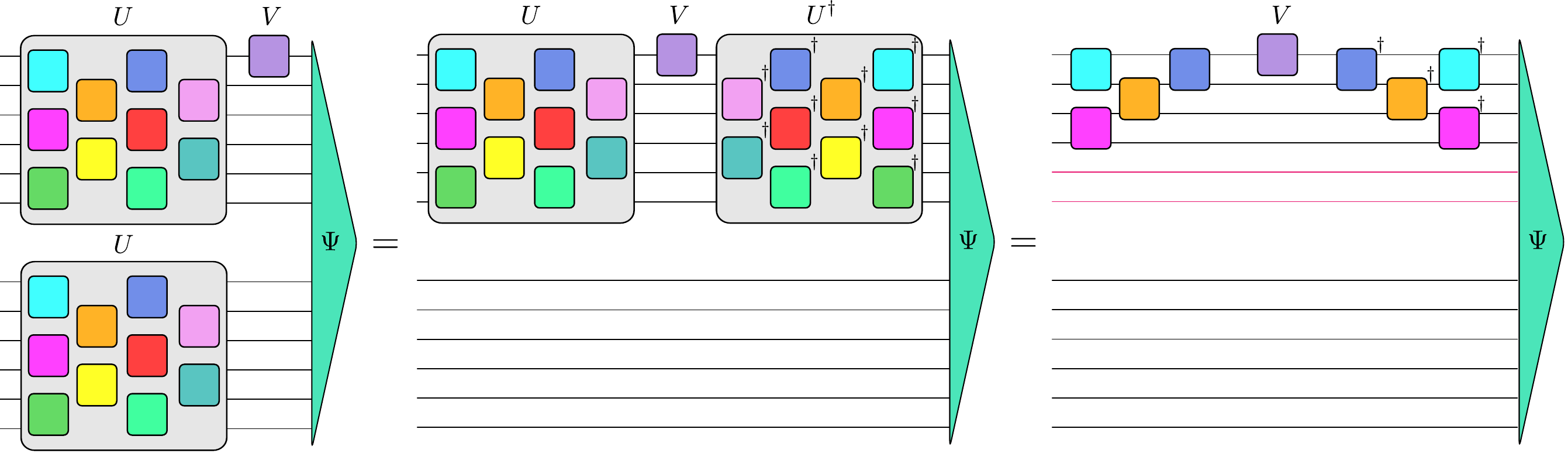}
    \caption{2-designs over a group $G$ which possesses an invariant state $\ket\Psi\in\mch\ts$ formed by averaging over an ensemble of sublinear-depth unitaries belonging to $G$ are straightforwardly disallowed  by a simple lightcone argument~\cite{west2025no}. The invariant state effectively acts as a resource for converting a unitary from the group to its Hermitian conjugate. Indeed, for $U\in G$ we have $U\ts (V\ot\id_\mch)\ket\Psi =   ((UVU\ad)\ot\id_\mch) \ket\Psi $, so that  there are regions of $(V\ot\id_\mch)$ that are ``undisturbed'' if (and with high probability only if) $U$ is of sublinear-depth. Modulo some mild technical assumptions on the entanglement structure of $\ket\Psi$ and the scrambling capacity of $G$, this can be used to rule out such sublinear-depth $G$-designs. } 
    \label{fig:1}
\end{figure}

\section{Results}
The following arguments are slight generalisations of those which may be found in Ref.~\cite{west2025no}; they bear also a considerable resemblance to those of its contemporary, Ref.~\cite{grevink2025will}. 
We begin by establishing some notation. 
Let $\mch$ be a finite-dimensional Hilbert space, $\mcl(\mch):={\rm End}\,\mch$, and let $\mce$ be an ensemble of unitaries on $\mch$ (a probability measure on $\mbu(\mch)$, if you like).  Its $k$\textsuperscript{th} moment channel is~\cite{mele2023introduction}
\begin{equation}
 \Phi_{\mce}^{(k)}(X)
 =
 \expect_{U\sim\mce} \Ad_{U\tk}(X)\,,
\end{equation}
where $\Ad_U(A):=UAU\ad$. 
For a (compact) group $G\subseteq\mbu(\mch)$, then, we write $\Phi_{\mu_G}^{(k)}$ for the
channel obtained from its Haar measure $\mu_G$.  We call $\mce$ (which we do not require to satisfy $\supp(\mce)\subseteq G$) an $\varepsilon$-approximate  $k$-design over $G$ when
\begin{equation}
 \big\|\Phi_{\mce}^{(k)}-\Phi_{\mu_G}^{(k)}\big\|_\diamond\leq\varepsilon\,,
\end{equation}
where we recall that the \textit{diamond distance}~\cite{benenti2010computing}   is defined by 
\begin{equation}
\|\phi-\phi'\|_\diamond = \sup_{\rho} \|((\phi-\phi')\otimes\id_{\rm anc})\rho\|_1, 
\end{equation}
where the supremum is taken over all  quantum states $\rho$ on $\HC\tk\otimes \HC_{\rm anc}$, with $\HC_{\rm anc}$ an arbitrarily large ancilla space. 
The diamond norm is operationalised by its close relationship to the one-shot distinguishability between two channels. Indeed, let $\mu,\nu$ be ensembles of unitaries on $\mch$, $\rho$ a state on $\mch\tk$, and $\{\Pi, \id_\mch\tk-\Pi\}$ a POVM on $\mch\tk$. Let $U$ be with probability 1/2 sampled from $\mu$, and with probability 1/2 sampled from $\nu$.   Then, upon acting with $U\tk$ on $\rho$ and measuring with respect to the POVM,  one has~\cite{west2025no}
\begin{equation}
   \big\|\Phi_\mu^{(k)}-\Phi_\nu^{(k)}\big\|_\diamond\geq 2\,\big\lvert p(\Pi|U\sim\mu)-p(\Pi|U\sim\nu) \big\rvert\,.\label{eq:dd}
\end{equation}
Finally, we say that an operator $O\in\mcl(\mch\tk)$ is \textit{$\mce$-invariant} if $\Phi_{\mce}^{(k)}(O)=O$.
Our primary observation is:

\begin{theorem}
\label{thm:main}
 Let $\mu$ and $\nu$ be ensembles of unitaries,   $Q\in\mcl(\mch\tk)$  a (nonzero) $\mu$-invariant  orthogonal projector, and set $\rho_Q=Q/\Tr(Q)$.  Fix a probe unitary
$A\in\mbu(\mch^{\otimes k})$.  
 Then, for any orthogonal projector $\Pi$ such that $\Ad_{\Ad_{U\tk}(A)}(Q)\leq\Pi$ for all $U\in\supp(\nu)$, we have:
\begin{equation}
 \big\|\Phi_\mu^{(k)}-\Phi_\nu^{(k)}\big\|_\diamond\geq 1-\Tr\left[ \Pi\,\Phi_\mu^{(k)}\left(\Ad_A(\rho_Q) \right)\right]\,.
 \label{eq:bound}
\end{equation}

\end{theorem}

\begin{proof}
Let $\delta=\|\Phi_\mu^{(k)}-\Phi_\nu^{(k)}\|_\diamond$,  $f=\expect_{U\sim\nu}
 \Tr \left[Q\Ad_{U\tk}(\rho_Q)\right]$ and   $f'=\Tr\big[\Pi\,\Phi_\mu^{(k)}\left(\Ad_A(\rho_Q)\right)\big]$. Recall that we can get lower bounds on the diamond distance $\|\Phi_\mu^{(k)}-\Phi_\nu^{(k)}\|_\diamond$ by considering (thought) experiments that would distinguish between the application of $U\tk$ for $U\sim\mu$ and for $U\sim\nu$ (applied to any initial state of our choosing). Let us begin by    choosing $\rho_Q$. 
  By the $\mu$-invariance of $Q$, the POVM $\{Q,\id_\mch\tk-Q\}$ accepts with probability one if $U\sim\mu$; if  $U\sim\nu$ the acceptance is with probability $f$. Eq.~\eqref{eq:dd} then immediately yields $\delta\geq 2(1-f)$.
  On the other hand, consider applying the channels to the initial state $\Ad_A(\rho_Q) $, and measuring with respect to $\{\Pi,\id_\mch\tk-\Pi\}$. 
By assumption, we have $\Ad_{\Ad_{U\tk}(A)}(Q)\leq\Pi$, so that
\begin{equation}
 \Tr\big(\Pi [\Ad_{U\tk}(\Ad_A(\rho_Q))]\big)\geq\Tr\big([\Ad_{\Ad_{U\tk}(A)}(Q)][\Ad_{U\tk}(\Ad_A(\rho_Q))]\big)
 =\Tr(Q\Ad_{U\tk}(\rho_Q)).
\end{equation}
Thus in this case the $\nu$-acceptance probability is at least $f$, while
the $\mu$-acceptance probability is
$f'$.  Eq.~\eqref{eq:dd} therefore gives $\delta\geq2(f-f')$. Combining our two lower bounds on $\delta$ immediately yields the claim of the theorem.
 
\end{proof}
For example, take $\mu$ to be the Haar measure on some group $G$. The two ``experiments'' (i.e., choices of POVM) of the proof of Theorem~\ref{thm:main} correspond to the previously foreshadowed dichotomy faced by a putative shallow $G$-design $\nu$. Indeed, if (the $k$-fold tensor action of) $U\sim \nu$ fails to approximately stabilise $Q$, then the POVM $\{Q,\id-Q\}$ will witness a large diamond distance from $\mu$; if it does not so fail, then the ``lightcone'' POVM  $\{\Pi,\id-\Pi\}$ will instead serve that purpose.\\

Geometrically, we can think of $\Pi$ as a sort of  local enlargement of the invariant subspace.  Indeed, suppose we have some contiguous bipartition $\mch=\mch_L\otimes\mch_{\overline L}$, and take $A=V\otimes\id_\mch^{\otimes(k-1)}$ for some local unitary $V\in\mbu(\mch)$.
Now, let us introduce
\begin{align*}
 \mathcal K_L(Q)&=\operatorname{span}
 \left\{\left[(O_L\otimes\id_{\overline L})\otimes\id^{\otimes(k-1)}_\mch\right]\ket{\xi}:O_L\in\mcl(\mch_L),\ \ket{\xi}\in\operatorname{ran}(Q)\right\},\\
 \Pi_L(Q)&=\operatorname{proj} \left(\mathcal K_L(Q)\right).
\end{align*}
If every $U\in{\rm supp}(\nu)$ satisfies $\Ad_U(V)\in \mbu(\mch_L)\otimes\id_{\overline L}$ (that is, a ``shallowness assumption'' on $\nu$) then we have that  $\Ad_{U\tk}(A)\operatorname{ran}(Q)
 \subseteq\mathcal K_L(Q)$, so that 
  $\Ad_{\Ad_{U\tk}(A)} (Q) \leq\Pi_L(Q)$, which is exactly the  hypothesis on the projector $\Pi$ of Theorem~\ref{thm:main}. 

\begin{remark}
  The preceding arguments also immediately apply to the case of ancillas. Indeed, suppose that each sampled unitary $W$ acts  on $\mch\otimes\mch_{\rm anc}$, with the
  ancillas initialised in an arbitrary fixed state
  $\rho_{\rm anc}$. Extending the POVMs by the identity on $\mch_{\rm anc}^{\otimes k}$, the argument of Theorem~\ref{thm:main}  goes through unchanged. 
  \end{remark}

All of this becomes quite concrete when we turn to our examples. Indeed, for the matchgate, orthogonal, and symplectic applications, we can take $k=2$ and $Q=\ketbra{\Psi}{\Psi}$, for a maximally entangled, $G$-invariant $\ket\Psi$. 
We write $\mch=\mch_L\otimes\mch_{\overline L}$, with dimensions $d_L$ and
$d_{\overline L}$, and let $d=d_Ld_{\overline L}$. The maximal
entanglement of $\ket\Psi$ means that  the induced vectorisation map $J_\Psi(O)=(O\otimes\id)\ket{\Psi}$
is an isometry when we equip  operator space with  the (normalised)
Hilbert--Schmidt inner product,
$\langle O_1,O_2\rangle_{\HS}=d^{-1}\Tr \big(O_1^\dagger O_2\big)$. These ingredients combine to give us a fairly  canonical ``enlargement'', namely 
\begin{equation}
    \Pi_L(\ketbra{\Psi}{\Psi})=\operatorname{proj}(J_\Psi(
\mcl(\mch_L)\ot \id_{\overline L}))\,;
\end{equation}
the corresponding (Hilbert--Schmidt) projection in operator space is 
\begin{equation}
 \mcp_L(O)= (J_\Psi\ad \Pi_L J_\Psi)(O)=\frac{1}{d_{\overline L}}\Tr_{\overline L}(O)\otimes\id_{\overline L}
 \label{eq:plo}
\end{equation}
(note the disappearance of the dependence on $\Psi$). 
We now specialise Theorem~\ref{thm:main} to the case of a pure state invariant in $\mch\ts$: 
\begin{theorem}\label{thm:thm2}
    Let $G\subseteq\mbu(\mch)$ be compact and possess a $G$-invariant state $Q=\ketbra{\Psi}{\Psi}\in\mcl(\mch\ts)$, with $\ket\Psi$ maximally entangled, and decompose $\mch=\mch_L\otimes\mch_{\overline L}$.  Take
$V\in\mbu(\mch)$, and let $\mce$ be any probability measure on
$\mbu(\mch)$ satisfying $\Ad_U(V)\in\mbu(\mch_L)\ot\id_{\overline{L}}$  for every
$U\in\supp(\mce)$. Then,  defining
\begin{equation}
 q_G(V,L)=\expect_{W\sim\mu_G}\frac{1}{d}
 \Tr \left[ \Ad_W(V)\ad \mcp_L  (\Ad_W(V))\right]\,,
 \label{eq:q-group}
\end{equation}
we have
\begin{equation}
 \big\|\Phi_\mce^{(2)}-\Phi_{\mu_G}^{(2)}\big\|_\diamond\geq 1-q_G(V,L)\,.
 \label{eq:thm2}
\end{equation}
\end{theorem}

\begin{proof}
We apply Theorem~\ref{thm:main} with
$Q=\ketbra{\Psi}{\Psi}$, $A=V\otimes\id$, $\Pi=\Pi_L(\ketbra{\Psi}{\Psi})$, $\mu=\mu_G$, and
$\nu=\mce$ (which by the previous discussion satisfy the various hypotheses of Theorem~\ref{thm:main}). This yields   
\begin{align}
\big\|\Phi_{\mu_G}^{(2)}-\Phi_\nu^{(2)}\big\|_\diamond &\geq 1-\Tr \left[\Pi_L\Phi_{\mu_G}^{(2)}\left((V\ot\id_\mch)\ketbra{\Psi}{\Psi}(V\ad\ot\id_\mch) \right)\right]\\
&= 1-\expect_{W\sim\mu_G}\Tr \left[\Pi_L  (WV\ot W)\ketbra{\Psi}{\Psi}(V\ad W\ad \ot W\ad) \right]\\
&= 1-\expect_{W\sim\mu_G}\Tr \left[\Pi_L  (\Ad_W(V) \ot\id_\mch)\ketbra{\Psi}{\Psi}(\Ad_W(V\ad)  \ot\id_\mch) \right]\\
&= 1-\expect_{W\sim\mu_G}\langle J_\Psi(\Ad_W(V)), \Pi_L  J_\Psi(\Ad_W(V))\rangle\,.
\end{align}
Now, for any unitary $U\in\mbu(\mch)$, we have 
\begin{equation}
    \langle J_\Psi(U), \Pi_L  J_\Psi(U)\rangle= \langle J_\Psi(U), J_\Psi J_\Psi\ad\Pi_L  J_\Psi(U)\rangle= \langle J_\Psi(U), J_\Psi \mcp_L  (U)\rangle\,= \frac{1}{d}\Tr(U\ad \mcp_L  (U)),
\end{equation}
where we have used that $J_\Psi$ is an isometry (with respect to the (normalised) Hilbert–Schmidt inner product). We therefore conclude that
\begin{align}
\big\|\Phi_\mce^{(2)}-\Phi_{\mu_G}^{(2)}\big\|_\diamond &\geq 1-\frac{1}{d}\expect_{W\sim\mu_G}\Tr(\Ad_W(V)\ad \mcp_L  (\Ad_W(V))) = 1-q_G(V,L)\,,
\end{align}
as desired.
\end{proof}

We note that if  $\supp(\mce)\subseteq G$ then it is possible to strengthen the lower bound  to
$2(1-q_G(V,L))$~\cite{west2025no}; the more general bound here comes thus at the cost of a factor of two. In the case of a qubit system, one can (by expanding $O$ in the Pauli basis) reformulate Eq.~\eqref{eq:thm2} to read
\begin{equation}
    \big\|\Phi_\mce^{(2)}-\Phi_{\mu_G}^{(2)}\big\|_\diamond \geq 1-\expect_{W\sim\mu_G}F_{\overline L}(\Ad_W(V))\,\label{eq:otocbound}\,,
\end{equation}
where
\begin{equation}
  F_{\overline L}(O)=\frac{1}{d_{\overline L}^{\,2}}\sum_{P\in\mathsf P_{\overline L}}\frac{1}{d}\Tr\!\left[O^\dagger(\id_L\otimes P)O(\id_L\otimes P)\right]\label{eq:otoc}
\end{equation}
is the infinite-temperature out-of-time-order correlator (OTOC)~\cite{roberts2017chaos,nahum2018operator} averaged over the Pauli probes $P$ with support on $\mch_{\overline L}$, with  $\mathsf P_{\overline L}$  the Pauli group (modulo phases) on $\mch_{\overline L}$. Up to the factor of two lost at this level of generality, Eq.~\eqref{eq:otocbound} is the result of Theorem 1 of Ref.~\cite{west2025no}.

\section{Applications}
Armed with Eq.~\eqref{eq:otocbound}, our applications to the matchgate, orthogonal and symplectic groups are essentially immediate, as the relevant OTOCs have already been calculated in Ref.~\cite{west2025no}; we deal with these cases fairly quickly. The case of the Clifford group requires the full Theorem~\ref{thm:main}, and thus marginally more attention.

\subsection{The matchgate group}
Let us write 
\begin{alignat}{3}\label{eq:majos}
c_1&=XIII\cdots I,\qquad&c_{2}&=YIII\cdots I\nonumber\\
c_3&=ZXII\cdots I,&c_{4}&=ZYII\cdots I\\
&\:\:\vdots&&\:\:\vdots\nonumber\\
c_{2n-1}&=ZZ\cdots ZX,&c_{2n}&=ZZ\cdots ZY  \nonumber
\end{alignat} 
for Jordan--Wigner Majorana operators corresponding to $n$ qubits on a line; the matchgate group~\cite{knill2001fermionic}  $\mathsf{MG}\cong{\rm Spin}(2n)$ is then generated by exponentials of quadratic  Hamiltonians in the Majoranas.  It possesses a maximally entangled invariant state of
the form $ \ket{\Psi_{\mathsf{MG}}}=(\id\otimes\Omega)\ket{\Phi}$, where $\ket{\Phi}\in\mch\ts$ is the Bell state, and $\Om$ may be taken to be either $XYXYX\ldots $ or $YXYXY\ldots $~\cite{west2025no}. 
Assuming for simplicity that $n$ is even (the odd case introduces no particularly interesting new features) consider the perturbation $V=X_{n/2}$.  If $U$ is
implemented by a depth-$L$ circuit of arbitrary nearest-neighbour
two-qubit gates and $L\leq n/2-1$, then the (Heisenberg) light cone of
$X_{n/2}$, i.e. an $X$ operator on the $(n/2)$\textsuperscript{th}-qubit, cannot reach the last qubit.  Thus $\Ad_U(V)\in \mbu(\mch_L)\otimes\id_{\overline L}$
holds with $\mch_L$ the first $n-1$ qubits and
$\mch_{\overline L}$ the last qubit (regardless of whether or not $U$ is a
matchgate).
It remains to evaluate the average OTOC of Eq.~\eqref{eq:otoc}. The benefit of choosing the perturbation to be 
  $X_{n/2}$ is that it lies in the  module (under the adjoint action) consisting of products of
$n-1$ distinct Majoranas, of size $\mco(\exp (n))$, which leads to a favourable OTOC scaling. Indeed one finds~\cite{west2025no} 
\begin{equation}
 \expect_{W\sim\mu_{\mathsf{MG}}}F_{\overline{L}}(\Ad_W(X_{n/2}))=\frac{\binom{2n-2}{n-1}}{\binom{2n}{n-1}}=\frac{n+1}{2(2n-1)}\,.
\end{equation}
 We therefore immediately obtain:
\begin{crl}
\label{cor:matchgate}
Let $\mce_L$ be any ensemble of $n$-qubit unitaries implementable by depth-$L$ nearest-neighbour circuits on a
one-dimensional line.  Then,  for even $n$ and $L\leq n/2-1$,
\begin{align}
 \big\|\Phi_{\mce_L}^{(2)}-\Phi_{\mu_{\mathsf{MG}}}^{(2)}\big\|_\diamond
 &\geq1-\frac{n+1}{2(2n-1)}\sim\frac34.
\end{align}
\end{crl}
In the case of odd $n$ we can take the  perturbation to be 
  $X_{(n-1)/2}$ and obtain a similar result.
So, for every fixed $\varepsilon<3/4$, we conclude that for sufficiently large $n$ any $\varepsilon$-approximate matchgate 2-design 
requires depth at least $n/2$ in the one-dimensional nearest-neighbour circuit model. The scaling is tight up to a (small!) constant, as  the full matchgate Haar measure
can be sampled from exactly using depth $3n$~\cite{braccia2025optimal}.

\subsection{The orthogonal and symplectic groups}

Let $d=2^n$ and $\mch=(\mathbb C^2)^{\otimes n}$.  We use the defining
representations
\begin{align}
 \mbo (d)
 &=
 \left\{W\in\mbu(d):W^{\mathsf T}W=\id\right\}\,,\\
 \mbsp(d/2)
 &=
 \left\{W\in\mbu(d):W\Omega_{\mbsp} W^{\mathsf T}=\Omega_{\mbsp}\right\}\,,
\end{align}
where $\Omega_{\mbsp}=iY\ot\id_{2^{n-1}}$ is the canonical symplectic form. 
The  Bell state $\ket{\Phi_d}$ is invariant under
$O^{\otimes2}$ for $O\in \mbo(d)$, and $\ket{\Psi_\mbsp}=(\id\otimes\Omega_\mbsp)\ket{\Phi_d}$ is invariant under the diagonal action of the symplectic group.  
Both of these states are maximally entangled, so we can apply Theorem~\ref{thm:thm2}.
In both cases, we take (as in Ref.~\cite{west2025no}) the   perturbing unitary to be $V=Z_1$, so that  a depth-$L$
nearest-neighbour circuit 
can spread $V$ over at most the first $L+1$ qubits. Taking $\mch_L$ to be the corresponding Hilbert space, the Weingarten calculus on the orthogonal and symplectic groups~\cite{collins2006integration} allows for the OTOCs of Eq.~\eqref{eq:otoc} to be readily evaluated, resulting in
\begin{align}
 \expect_{W\sim\mu_{\mbo}}F_{\overline{L}}(\Ad_W(Z_{1}))&=\frac{d_L^2+d_L-2}{(d+2)(d-1)}\\
 \expect_{W\sim\mu_{\mbsp}}F_{\overline{L}}(\Ad_W(Z_{1}))&=\frac{d_L(d_L+1)}{d(d+1)}\,. 
\end{align}
We therefore find:

\begin{crl}
\label{cor:orthogonal-symplectic}
Let $\mce_L$ be any ensemble of $n$-qubit unitaries implementable by depth-$L$ nearest-neighbour circuits on a
one-dimensional line.  For $n\geq2$ and
$0\leq L\leq n-2$,
\begin{align}
 \left\|\Phi_{\mce_{L}}^{(2)}-\Phi_{\mu_{\mbo(d)}}^{(2)}\right\|_\diamond&\geq\frac{d(3d+2)}{4(d+2)(d-1)}\sim \frac34,\\
 \left\|\Phi_{\mce_{L}}^{(2)}-\Phi_{\mu_{\mbsp(d/2)}}^{(2)}\right\|_\diamond&\geq\frac{3d+2}{4(d+1)}\sim\frac34.
\end{align}
\end{crl}
So, for any fixed $\varepsilon<3/4$ (and  sufficiently large $n$) orthogonal and symplectic designs require depth at least $n-1$.

\subsection{The Clifford group}
\label{sec:clifford}
The famously narrow failure of the $n$-qubit Clifford group $\mathsf{Cl}_n$ to form a unitary 4-design~\cite{zhu2016clifford,bittel2025complete} arises as a result of the element
\begin{equation}
 Q=\frac{1}{d^2}
 \sum_{P\in\mathsf{P}_n}
 P^{\otimes4}
 \label{eq:cliffordq}
\end{equation}
of the 4\textsuperscript{th}-order Clifford commutant. Indeed,  Clifford conjugation permutes the Pauli group up to signs, which disappear in the fourth tensor power. One can readily verify that $Q$ is an orthogonal projector of trace $d^2$ (recall $d=2^n$); we will use it to apply Theorem~\ref{thm:main}. We find:

\begin{crl}
Let $\mce_L$ be any probability ensemble of $n$-qubit unitaries implementable by  depth-$L$ nearest-neighbour circuits on a one-dimensional line.   For $0\leq L\leq n-1$,
\begin{equation}\label{eq:cliffbound}
 \big\|\Phi_{\mce_L}^{(4)}-\Phi_{\mu_{\mathsf{Cl}_n}}^{(4)}\big\|_\diamond\geq1-\frac{4^{L+1}-1}{4^n-1}\,.
\end{equation}
In particular, every ensemble of depth $L\leq n-2$ has distance at
least $3/4$ from the Clifford fourth moment. 
\end{crl}

\begin{proof}
Take $V=Z_1$, i.e. the $Z$ operator on the first of $n$ qubits,   $A=Z_1\otimes\id_\mch^{\otimes3}$, and the decomposition $\mch=\mch_L\otimes\mch_{\overline L}$, with $\mch_L$ the Hilbert space of the first $L+1$ qubits.  Notice that (after regrouping replicas by spatial subsystem) Eq.~\eqref{eq:cliffordq} factorises as $Q=Q_L\otimes Q_{\overline L}$.
We consider the POVM $\{\Pi_L, \id_{\mch^{\ot4}}-\Pi_L\}$, with $\Pi_L=\id_{\mch_L^{\otimes4}}\otimes Q_{\overline L}$. 
For any $U\in\supp(\mce_L)$, write
$\Ad_U(V)=(\Ad_U(V))\rvert_{\mch_L}\otimes\id_{\overline L}$; up to some rearranging of tensor factors one then manifestly has 
\begin{equation}
\Ad_{\Ad_{U^{\ot4}}(A)}(Q)=\left[((\Ad_U(V))\rvert_{\mch_L}\otimes\id^{\otimes3}_\mch)Q_L(((\Ad_U(V))\rvert_{\mch_L})^\dagger\otimes\id^{\otimes3}_\mch)\right]\otimes Q_{\overline L}\leq \Pi_L\,,
\end{equation}
so that Theorem~\ref{thm:main} applies.

It remains to calculate the acceptance probability $\Tr\left[ \Pi_L\,\Phi_{\mu_{\mathsf{Cl}_n}}^{(4)}\left(\Ad_A(\rho_Q) \right)\right]$, with $\rho_Q=Q/d^2$. Unravelling some definitions, we have
\begin{align}
    \Phi_{\mu_{\mathsf{Cl}_n}}^{(4)}\left(\Ad_A(\rho_Q) \right)&= \frac{1}{d^2}\expect_{C\sim\mathsf{Cl}_n}C^{\ot 4} AQA\ad (C\ad)^{\ot 4}\\
    &= \frac{1}{d^2}\expect_{C\sim\mathsf{Cl}_n}((CZ_1)\ot C^{\ot 3}) Q ((Z_1C\ad)\ot (C\ad)^{\ot 3}) \\
    &= \frac{1}{d^2}\expect_{C\sim\mathsf{Cl}_n}((CZ_1C\ad)\ot\id^{\ot 3}_\mch) Q ((CZ_1C\ad)\ot\id^{\ot 3}_\mch) \,,
\end{align}
where in the last line we have used the Clifford-invariance of $Q$. 
Now, conjugation maps $Z_1$ uniformly over the $d^2-1$ nonidentity Pauli
classes (up to an irrelevant sign).  Let
$P=P_L\otimes P_{\overline L}$ be one such Pauli.  If
$P_{\overline L}=\id_{\overline{L}}$, then $\Pi_L$ accepts with probability one.  If
$P_{\overline L}\neq\id_{\overline{L}}$, choose a Pauli
$S_{\overline L}$ anticommuting with it. As one can freely absorb a factor of 
$S_{\overline L}^{\otimes4}$ in and out of  
$Q_{\overline L}$, we have
\begin{equation}
    Q_{\overline L}
 (P_{\overline L}\otimes\id^{\otimes3}_{\mch_{\overline L}})
 Q_{\overline L} = Q_{\overline L}(S_{\overline L}^{\otimes4})
 (P_{\overline L}\otimes\id^{\otimes3}_{\mch_{\overline L}})
 Q_{\overline L} = -Q_{\overline L}
 (P_{\overline L}\otimes\id^{\otimes3}_{\mch_{\overline L}})(S_{\overline L}^{\otimes4})
 Q_{\overline L} = -Q_{\overline L}
 (P_{\overline L}\otimes\id^{\otimes3}_{\mch_{\overline L}})
 Q_{\overline L} =0 \,.
\end{equation} 
So in this case,  the acceptance probability is zero. As exactly $d_L^2-1$ of the $d^2-1$ nonidentity Paulis are supported
entirely on $L$, we conclude
\begin{equation}
 \Tr\left[ \Pi_L\,\Phi_{\mu_{\mathsf{Cl}_n}}^{(4)}\left(\Ad_A(\rho_Q) \right)\right]=\frac{d_L^2-1}{d^2-1}=\frac{4^{L+1}-1}{4^n-1}.
\end{equation}
Theorem~\ref{thm:main} now immediately gives
Eq.~\eqref{eq:cliffbound}.
\end{proof}

\section{Discussion}
The Haar measure on the full unitary group is a significantly more complicated object than those of  the matchgate and Clifford groups, characterised as it is by exponentially more parameters. In the cases of the orthogonal and symplectic   groups the separations are  less dramatic, though  they are nonetheless strictly simpler objects from which exactly sampling is easier. The fact, then,  that when one settles for an approximation to the various Haar measures in the form of an approximate design the unitary case becomes (by far) the easiest is  therefore something of a counterintuitive finding. In the previous works addressing this question~\cite{west2025no,grevink2025will}, it was possible to see the increased difficulty as a result of something of a trade-off: In those works, only ensembles of unitaries consisting of elements that themselves belonged to the respective subgroups were considered, so that the true difficulty of approximating designs over those groups was obscured. Indeed, on the one hand one was trying to approximate a simpler object than the full unitary measure, but on the other hand, one was constrained to use only a restricted set of unitaries. Our finding here, that even when allowing ambient unitaries from beyond the group the approximation of such designs is exponentially harder than the unitary case, is then particularly striking. \\

It is natural to wonder whether the lower bounds of order $\Om(n)$ on the depths required to construct approximate designs in each of our considered scenarios are tight; in fact, they are. First, it is known that arbitrary Clifford gates may be synthesised in linear depth~\cite{maslov2007linear}, establishing the tightness of our bound in the Clifford case. Next, and interestingly, considering the intersection of the Cliffords with our three other example groups are   sufficient to produce   3-designs in each case~\cite{wan2022matchgate,gargiulo2026pauli}, which inherit the linear depth of the Clifford group. The matchgate case is perhaps particularly remarkable, as   Haar-random  matchgates can be synthesised in depth $3n$~\cite{braccia2025optimal}, giving a constant factor separation between depths at which approximate 2-designs are impossible, and the depth at which the Haar measure can be sampled from \textit{exactly}.\\

Throughout this work we have been implicitly using the ``defining'' representations of our considered groups, i.e., taking the matrices to represent themselves. In principle, however, we can repeat the story with any representation.
A further   implication of our results is then the constraints they can place on the depths of unitary transforms which map (equivalent) representations of a group to one another.  To be concrete, suppose that we have two representations of a group $G$, i.e. $\pi:G\to {\rm End}\,(\mch_\pi)$ and $\sg:G\to {\rm End}\,(\mch_\sg)$, along with a unitary  \textit{intertwiner} $\Lm:\mch_\pi\to  \mch_\sg$, such that for every $g\in G$ we have $(\sg(g))\circ \Lm=\Lm\circ (\pi(g))$. If we have a lower bound $d_\sg$ on the  depth required to sample from a design when working in the  representation $\sg$, then, we can immediately conclude that $d_\pi\geq d_\sg - 2\,{\rm depth}(\Lm) $, where ${\rm depth}(\Lm)$ is the circuit depth needed to implement $\Lm$ using local nearest-neighbour gates. Indeed,   this bound follows from $\sg(g) = \Lm \pi(g) \Lm\ad$, so that one way of generating random matrices according to $\sg$ is to transform to $\pi$, construct the matrices there, and transform back.\\

Beyond their status as   intriguing mathematical statements, our results have significant implications for the practicality of protocols which rely on sampling from approximate designs over the groups considered here. Indeed, many protocols for both characterisation~\cite{hashagen2018real,helsen2022matchgate,chapman2026fermionic} and learning~\cite{west2026particle,zhao2021fermionic,wan2022matchgate,west2026classical,heyraud2024unified,west2025real,low2022classical,west2024random,heyraud2024unified} rely on sampling from these groups. To take a particular example, there is considerable interest in performing classical shadows~\cite{huang2020predicting} over the matchgate group, for which one requires an (approximate) matchgate 3-design~\cite{zhao2021fermionic,west2026particle,wan2022matchgate,west2026classical,low2022classical,heyraud2024unified}. 
While the   ruling out of shallow designs over these groups under the previously considered restricted scenarios~\cite{west2025no,grevink2025will} was interesting, there remained the possibility of constructing them some other way, thus allowing for a sublinear-depth implementation of those protocols. Unfortunately,  our results   rule out such a   happy occurrence.

\section{Acknowledgements}
MW, MC and ML acknowledge support by the Laboratory Directed Research and Development (LDRD) program of LANL under project number 20260043DR, and by LANL’s ASC Beyond Moore’s
Law project. This work was also supported by the Quantum Science Center (QSC), a National Quantum Information Science Research Center of the U.S. Department of
Energy (DOE). We acknowledge useful conversations with GPT-5.6. 

\bibliography{quantum}

\end{document}